\documentclass[conference,letterpaper,twoside]{IEEEtran}

\usepackage[cmex10]{amsmath}
\usepackage{graphicx,epic,eepic,epsfig,latexsym,amssymb,verbatim,color,revsymb}

\newcommand{\abs}[1]{\left\lvert#1\right\rvert} % absolute value: single vertical bars
\usepackage[english]{babel}
\usepackage[T1]{fontenc}
\usepackage{mathrsfs}
\usepackage{enumerate}
\usepackage{graphicx}
\usepackage{mathtools}
\usepackage[export]{adjustbox}
\usepackage{microtype}
\usepackage{MnSymbol}
\usepackage{lipsum}

\usepackage[colorlinks = true,
            linkcolor = red,
            urlcolor  = blue,
            citecolor = blue,
            anchorcolor = red]{hyperref}

\newtheorem{theorem}{Theorem}
\newtheorem{corollary}[theorem]{Corollary}

\newtheorem{lemma}[theorem]{Lemma}
 
\newtheorem{definition}[theorem]{Definition}  
\newtheorem{remark}[theorem]{Remark}  

\newcommand\qedsymbol{$\blacksquare$}
\newcommand\qed{\hfill\qedsymbol}

\newlength{\blank}
\newenvironment{proof-of}[1][{\hspace{-\blank}}]{{\medskip\noindent\textit{Proof~{#1}.\ }}}{\hfill\qedsymbol}
\newenvironment{proof}{{\medskip\noindent\textit{Proof.\ }}}{\hfill\qedsymbol}

\newcommand{\Tr}{{\operatorname{Tr}\,}}

\newcommand{\id}{{\operatorname{id}}}

\newcommand{\1}{\openone}
\newcommand{\ket}[1]{|#1\rangle}
\newcommand{\bra}[1]{\langle #1|}

\newcommand{\proj}[1]{|#1\rangle\!\langle #1|}

\newcommand{\cE}{{\mathcal{E}}}

\newcommand{\cX}{{\mathcal{X}}}
\newcommand{\cS}{{\mathcal{S}}}

\begin{document}

\title{Entanglement-Assisted Quantum Data Compression} 
%\title{Compression of Quantum Source Ensembles \protect\\ with Entanglement Assistance} 

\author{%
  \IEEEauthorblockN{Zahra Baghali Khanian\IEEEauthorrefmark{1}\IEEEauthorrefmark{2}}
  \IEEEauthorblockA{\IEEEauthorrefmark{1}%
                    ICFO\\
                    Barcelona Institute of Technology\\
                    08860 Castelldefels, Spain\\
                    Email: zbkhanian@gmail.com}
  \and
  \IEEEauthorblockN{ }
  \IEEEauthorblockA{\IEEEauthorrefmark{2}%
                    Grup d'Informaci\'{o} Qu\`{a}ntica\\
                    Departament de F\'{\i}sica\\
                    Universitat Aut\`{o}noma de Barcelona\\
                    08193 Bellaterra (Barcelona), Spain}
  \and
  \IEEEauthorblockN{Andreas Winter\IEEEauthorrefmark{2}\IEEEauthorrefmark{3}}
  \IEEEauthorblockA{\IEEEauthorrefmark{3}%
                    ICREA\\
                    Pg.~Lluis Companys, 23\\
                    08010 Barcelona, Spain\\
                    Email: andreas.winter@uab.cat}
}
\maketitle

\begin{abstract}
Ask how the quantum compression of ensembles of pure states is
affected by the availability of entanglement, and in settings where
the encoder has access to side information.
We find the optimal asymptotic quantum rate and the optimal tradeoff 
(rate region) of quantum and entanglement rates. 
It turns out that the amount by which the quantum rate beats
the Schumacher limit, the entropy of the source, is precisely half
the entropy of classical information that can be extracted from the
source and side information states without disturbing them at all
(``reversible extraction of classical information'').

In the special case that the encoder has no side information, or that 
she has access to the identity of the states, this problem reduces to the known
settings of \textit{blind} and \textit{visible} Schumacher compression, respectively,
albeit here additionally with entanglement assistance. 
We comment on connections to previously studied and further rate
tradeoffs when also classical information is considered.
\end{abstract}

\section{Quantum sources with side information} 
The task of data compression of a quantum source, introduced by Schumacher
\cite{Schumacher1995}, marks one of the foundations of quantum
information theory: not only did it provide an information theoretic 
interpretation of the von Neumann entropy $S(\rho) = -\Tr\rho\log\rho$
as the minimum compression rate, it also motivated the very concept of the qubit!
(Throughout this paper, $\log$ denotes by default the binary logarithm.)
In the Schumacher modelling, a source is given by an ensemble 
$\cE = \{ p(x), \proj{\psi_x} \}$ of pure states $\psi_x=\proj{\psi_x}\in\cS(A)$,
$\ket{\psi_x}\in A$, with a Hilbert space $A$ that in this paper we shall assume to be of
finite dimension $|A|<\infty$; $\cS(A)$ denotes the set of states (density operators).
Furthermore, $x\in\cX$ ranges over a 
discrete alphabet, so that we can can describe the source equivalently by
the classical-quantum (cq) state $\omega = \sum_x p(x) \proj{x}^X \otimes \proj{\psi_x}^A$.
%Hence we may describe the source equivalently via the
%random variable $X \in \cX$, distributed according to $p$, i.e. $\Pr\{X=x\}=p_x$;
%this also makes the pure state $\Psi = \psi_X$ a random variable.

While the achievability of the rate $S(A)_\omega = S(\omega^A)$ was
shown in \cite{Schumacher1995,Jozsa1994_1} (see also \cite[Thm.~1.18]{OhyaPetz:entropy}),
the full (weak) converse was established in \cite{Barnum1996}, a simplified
proof being given by M. Horodecki \cite{Horodecki1998}; the strong 
converse was proved in \cite{Winter1999}. 

\medskip
In this paper, we consider a more comprehensive model, where on the one hand
the sender/encoder of the compressed data (Alice) has access to side
information, namely a pure state $\sigma_x^C$ in addition to the source state
$\psi_x^A$, and on the other hand, she and the receiver/decoder of the compressed
data (Bob) share pure state entanglement in the form of EPR pairs at a
certain rate.

Thus, the source is now an ensemble $\cE = \{ p(x), \proj{\psi_x}^A\otimes\proj{\sigma_x}^C \}$
of product states, which can be described equivalently by the cqq-state
\begin{align}
  \label{eq:source state omega}
  \omega^{XAC}=\sum_{x\in \mathcal{X}} p(x) \proj{x}^X\otimes \proj{\psi_x}^A\otimes \proj{\sigma_x}^C. 
\end{align}
Yet another equivalent description is via the
random variable $X \in \cX$, distributed according to $p$, i.e. $\Pr\{X=x\}=p_x$;
this also makes the pure states $\psi_X$ and $\sigma_X$ random variables.

We will consider the information theoretic limit of
many copies of $\omega$, i.e.~$\omega^{X^n A^n C^n} = \left(\omega^{XAC}\right)^{\otimes n}$:
\[
  \omega^{X^n A^n C^n}
    \!\!\! = \!\!\!\!
             \sum_{x^n \in \mathcal{X}^n}\!\!\!\! p(x^n) \proj{x^n}^{X^n} 
                       \!\otimes\! \proj{\psi_{x^n}}^{A^n}
                       \!\otimes\! \proj{\sigma_{x^n}}^{C^n}\!\!\!\!\!,
\]
using the notation
\begin{align*}
  x^n              &= x_1 x_2 \ldots x_n,\quad\;
  p(x^n) = p(x_1) p(x_2)  \cdots p(x_n), \\
  \ket{x^n}        &= \ket{x_1} \ket{x_2} \cdots \ket{x_n},\ 
  \ket{\psi_{x^n}} = \ket{\psi_{x_1}} \ket{\psi_{x_2}} \cdots \ket{\psi_{x_n}}.
\end{align*}

\medskip
\textbf{Further notation.}
Conditional entropy and conditional mutual information, $S(A|B)_{\omega}$ and $I(A:B|C)_{\omega}$,
respectively, are defined in the same way as their classical counterparts: 
\begin{align*}
  S(A|B)_{\omega}   &= S(AB)_\omega-S(B)_{\omega}, \text{ and} \\ 
  I(A:B|C)_{\omega} &= S(A|C)_\omega-S(A|BC)_{\omega} \\
                    &= S(AC)_\omega+S(BC)_\omega-S(ABC)_\omega-S(C)_\omega.
\end{align*}
%We use the letter $C$ also to denote a quantum system; this is not going to be 
%confusing if one notices that the von Neumann entropy is a function of quantum 
%states denoted as $S(\cdot)$. 
The fidelity between two states $\omega$ and $\xi$ is defined as 
\(
 F(\omega, \xi) = \|\sqrt{\omega}\sqrt{\xi}\|_1 
                = \Tr \sqrt{\omega^{\frac{1}{2}} \xi \omega^{\frac{1}{2}}},
\) 
with the trace norm $\|X\|_1 = \Tr|X| = \Tr\sqrt{X^\dagger X}$.

\section{Compression assisted by entanglement}
\label{sec:Compression assisted by entanglement}
We assume that the encoder, Alice, and the decoder, Bob, have initially a maximally 
entangled state $\Phi_K^{A_0B_0}$ on registers $A_0$ and $B_0$ (both of dimension $K$).
With probability $p(x^n)$, the source provides Alice  
with the state $\psi_{x^n}^{A^n}\otimes\sigma_{x^n}^{C^n}$.
Then, Alice performs her encoding operation 
$\mathcal{C}:A^nC^nA_0 \longrightarrow \hat{C}^nC_A$ on the systems $A^n$, 
$C^n$ and her part $A_0$ of the entanglement, which is a quantum channel,
i.e.~a completely positive and trace preserving (CPTP) map. 
(Note that our notation is a slight abuse, which we maintain as it is simpler 
while it cannot lead to confusions, since channels really are maps between
the trace class operators on the involved Hilbert spaces.)
The dimension of the compressed system obviously has to be smaller than the 
original source, i.e. $|C_A| \leq  \abs{A}^n$. 
We call $Q=\frac1n \log|C_A|$ and $E=\frac{1}{n}\log K$ the quantum and entanglement 
rates of the compression protocol, respectively.
The system $C_A$ is then sent to Bob via a noiseless quantum channel, who performs
a decoding operation $\mathcal{D}:C_A B_0 \longrightarrow \hat{A}^n$ on the system 
$C_A$ and his part of entanglement $B_0$. 

According to Stinespring's theorem \cite{Stinespring1955}, all these CPTP maps can be 
dilated to isometries 
$V_A : A^nC^nA_0 \hookrightarrow \hat{C}^n C_A W_A$ and  
$V_B : C_A B_0 \hookrightarrow {\hat{A}^n W_B}$, 
where the new systems $W_A$ and $W_B$ are the environment systems
of Alice and Bob, respectively. 

We say the encoding-decoding scheme has fidelity $1-\epsilon$, or error $\epsilon$, if 
\begin{align}
  \label{eq:Schumcaher assisted fidelity}
  \overline{F} &:= F\left( \omega^{X^n\hat{A}^n\hat{C}^n},\xi^{X^n\hat{A}^n\hat{C}^n} \right) \nonumber\\
               &=\sum_{x^n \in \mathcal{X}^n}\!\!\! p(x^n)F\!\left(\proj{\psi_{x^n}}^{A^n}\!
                          \otimes\!\proj{\sigma_{x^n}}^{C^n}\!,\xi_{x^n}^{\hat{A}^n\hat{C}^n}\right) \\
         &\geq 1-\epsilon, \nonumber
\end{align}
where $\xi^{X^n\hat{A}^n\hat{C}^n}=\sum_{x^n}p(x^n)\proj{x}^{X^n} \otimes \xi_{x^n}^{\hat{A}^n\hat{C}^n}$ 
and 
$\xi_{x^n}^{\hat{A}^n\hat{C}^n}=(\mathcal{D}\circ\mathcal{C})\!\proj{\psi_{x^n}\!}^{A^n} \!\otimes\! \proj{\sigma_{x^n}\!}^{C^n} \!\otimes\!\Phi_K^{A_0\!B_0}\!\!$.
We say that $(E,Q)$ is an (asymptotically) achievable rate pair if for all $n$
there exist codes such that the fidelity converges to $1$, and
the entanglement and quantum rates converge to $E$ and $Q$, respectively.
The rate region is the set of all achievable rate pairs, as a subset of 
$\mathbb{R}\times\mathbb{R}_{\geq 0}$. 

Note that this means that we demand not only that Bob can reconstruct the
source states $\psi_{x^n}$ with high fidelity on average, but that Alice
retains the side information states $\sigma_{x^n}$ as well with high fidelity.

There are two extreme cases of the side information that have been considered
in the literature: 
If $C$ is a trivial system, or more generally if the states
$\sigma_x^C$ are all identical, then the aforementioned task is the 
entanglement-assisted version of \textit{blind} Schumacher compression. If $C=X$, or more
precisely $\ket{\sigma_x}=\ket{x}$, then Alice has access to classical random variable 
$X$, and the task reduces to \textit{visible} Schumacher compression with entanglement assistance.  
The blind-visible terminology is originally from \cite{Barnum1996,Horodecki2000}.

\begin{remark}
\label{remark:E=0}
In the case of no entanglement being available, i.e. $E=0$ ($K=1$), the 
problem is fully understood: The asymptotic rate $Q=S(A)$ 
from \cite{Schumacher1995,Jozsa1994_1} is achievable without touching
the side informatiomn, and it is optimal, even in the visible case
(which includes all other side informations), by the weak and strong
converses of \cite{Barnum1996,Horodecki1998} and \cite{Winter1999}. 
\qed
\end{remark}

%In the present paper, we find the optimal compression rates in the general case of side 
%information $C$ and as a result for blind and visible Schumacher compression with 
%entanglement assistance. 

\section{Optimal quantum rate}
To formulate the minimum compression rate under unlimited entanglement
assistance, we need the following concept.

\begin{definition}
  \label{def:reducibility}
  An ensemble of pure states 
  $\cE=\{p(x),\proj{\psi_x}^A \otimes \proj{\sigma_x}^C \}_{x\in \mathcal{X}}$ 
  is called \emph{reducible} if its states fall into two or more orthogonal subspaces.
  Otherwise the ensemble $\cE$ is called \emph{irreducible}.
  We apply the same terminology to the source cqq-state $\omega^{X A C}$.
\end{definition}

Notice that a reducible ensemble can be written uniquely as a disjoint union 
of irreducible ensembles $\mathcal{E} = \bigcupdot_{y \in \mathcal{Y}} q(y) \mathcal{E}_y$,
with a partition $\mathcal{X} = \bigcupdot_{y\in \mathcal{Y}} \mathcal{X}_y$ and 
irreducible ensembles 
$\mathcal{E}_y = \{ p(x|y), \proj{\psi_x}^A \otimes \proj{\sigma_x}^C \}_{x \in \mathcal{X}_y}$,
where $q(y)p(x|y)=p(x)$ for $x \in \mathcal{X}_y$ and 
$q(y) =\sum_{x \in \mathcal{X}_y} p(x)$.
We define the subspace spanned by the vectors of each irreducible ensemble as 
$F_y := \text{span} \{\ket{\psi_x}\otimes\ket{\sigma_x} : x \in \mathcal{X}_y\}$.
The irreducible ensembles $\mathcal{E}_y$ are pairwise orthogonal, i.e.~$F_{y'} \perp F_y$ 
for all $y' \neq y$.
We may thus introduce the random variable $Y=Y(X)$ taking values in the set 
$\mathcal{Y}$ with probability distribution $q(y)$; namely, $Y$ is a deterministic
function of $X$ such that $\Pr\{X\in\mathcal{X}_Y\}=1$. 

We define the \textit{modified} source 
$\omega^{XACY}=\sum_x p(x) \proj{x}^X\otimes \proj{\psi_x}^A \otimes \proj{\sigma_x}^C \otimes \proj{y(x)}^Y$ 
with side information systems $CY$.
Because there is an isometry $V:AC \rightarrow ACY$ which acts as
\begin{equation}
  \label{eq:iso}
  V\ket{\psi_x}^A \otimes \ket{\sigma_x}^C=\ket{\psi_x}^A \otimes \ket{\sigma_x}^C \otimes \ket{y(x)}^Y,
\end{equation}
the extended source $\omega^{XACY}$ is equivalent to the original
source and side information $\omega^{XAC}$ modulo a local operation of Alice.

We first present the optimal asymptotic compression rate in the following 
theorem and prove the achievability of it, but we leave 
the converse proof to the end of this section, as it requires introducing 
further machinery.

\begin{theorem}
  \label{theorem: main}
  For the given source $\omega^{XACY}$, the optimal asymptotic compression rate 
  assisted by unlimited entanglement is
  $Q=\frac12 (S(A)+S(A|CY))$.
  
  Furthermore, there is a protocol achieving this communication 
  rate with entanglement consumption at rate $E=\frac12 (S(A)-S(A|CY))$.
\end{theorem}
\begin{proof}
We first show that this rate is achievable.
Consider the following purification of $\omega^{XACY}$,
\begin{align*}
  \ket{\omega}^{XX'ACY}=\sum_x \sqrt{p(x)}\ket{x}^X\ket{x}^{X'}\ket{\psi_x}^A\ket{\sigma_x}^C\ket{y(x)}^Y,
\end{align*}
with side information systems $CY$. This is obtained from 
$\ket{\omega}^{XX'AC}=\sum_x \sqrt{p(x)}\ket{x}^X\ket{x}^{X'}\ket{\psi_x}^A\ket{\sigma_x}^C$
by Alice applying the isometry $V$ from Eq.~(\ref{eq:iso}).

We apply quantum state redistribution (QSR) \cite{Devetak2008_2,Oppenheim2008} 
as a subprotocol, where the objective is for Alice to send to Bob $A^n$,
using $C^nY^n$ as side information, while $(XX')^n$ serves as reference system;
the figure of merit is the fidelity with the original pure state $(\omega^{XX'ACY})^{\otimes n}$.
Denoting the overall encoding-decoding CPTP map 
$\Lambda:A^nC^nY^n \rightarrow \hat{A}^n\hat{C}^n\hat{Y}^n$, 
QRS gives us the first inequality of the following chain:
\begin{align*} 
  1-o(1) &\leq F\!\left( \omega^{X^nX'^nA^nC^nY^n}\!\!,
                         (\id_{X^nX'^n} \otimes \Lambda) \omega^{X^nX'^nA^nC^nY^n}\! \right) \\
         &\leq F\!\left( \omega^{X^nA^nC^nY^n}\!\!,
                         (\id_{X^n} \otimes \Lambda) \omega^{X^nA^nC^nY^n}\! \right),
\end{align*}
where the second inequality follows from monotonicity of the fidelity under partial trace.
Thus, the protocol satisfies our fidelity criterion (\ref{eq:Schumcaher assisted fidelity}).

The communication rate we obtain from QSR is
$Q = \frac{1}{2}I(A:XX') = \frac{1}{2} (S(A)+S(A|CY))$. 
Furthermore, QSR guarantees entanglement consumption at the rate 
$E = \frac12 I(A:CY) = \frac12 (S(A)-S(A|CY))$.
\end{proof}
%Notice that a reducible ensemble can equivalently be described  using  a random variable $Y$, which is a function of random variable $X$ with probability distribution $p(x)$, that labels the orthogonal subspaces. Namely, the ensemble can be described by  $\cE'=\bigcupdot_y p(y) \mathcal{E}_y=\{p(x_y|y)p(y),\proj{\psi_{x_y}}  \}_{x_y,y}$  where $\mathcal{E}_y=\{p(x_y|y),\proj{\psi_{x_y}}\}_{x_y}$ and $\mathcal{E}_y \perp \cE_y$ for all $y\neq y'$.

\medskip
To prove optimality (the converse), we first need a few preparations.
The following definition is inspired by the ``reversible extraction of classical
information'' in \cite{Barnum2001_2}.

\begin{definition}
  \label{def:I_epsilon}
  For a source $\omega^{XAC}$ and $\epsilon \geq 0$, define
  \begin{align*}
    I_\epsilon(\omega) \!\!:= \!\!
                  \max_{V:AC\rightarrow \hat{A}\hat{C} W \text{ isometry}} \!\! \!\!I(X\!:\!\hat{C}W)_\xi 
          \text{ s.t. } F(\!\omega^{\!X\!A\!C\!}\!,\xi^{\!X\hat{A}\!\hat{C}\!})\! \geq \!1\!\!-\!\epsilon,    
\end{align*}
where  
\[
  \xi^{X\!\hat{A}\hat{C}W} \!\!\!
      =\!(\!\1_X \otimes V\!) \omega^{XAC}\!(\!\1_X \otimes V^{\dagger}\!) \!
      =\!\sum_x p(x) \proj{x}^X \!\otimes \proj{\xi_x}^{\!\hat{A}\hat{C}W} \!\!\!.
\]
\end{definition}
In this definition, the dimension of the environment is w.l.o.g. bounded as $|W| \leq |A|^2|C|^2$; 
hence, the optimisation is of a continuous function over a compact domain, so we have a 
maximum rather than a supremum.

\begin{lemma}
  \label{lemma:I_epsilon properties}
  The function $I_\epsilon(\omega)$ has the following properties:
  \begin{enumerate}
    \item It is a non-decreasing function of $\epsilon$. 
    \item It is concave in $\epsilon$.
    \item It is continuous for $\epsilon \geq 0$. 
    \item For any two states $\omega_1^{X_1 A_1 C_1}$ and $\omega_2^{X_2 A_2 C_2}$ and for $\epsilon \geq 0$,
          \(
            I_{\epsilon}(\omega_1 \otimes \omega_2) \leq  I_{\epsilon}(\omega_1) +I_{\epsilon}(\omega_2).
          \)
    \item For any state $\omega^{XAC}$, $I_0(\omega) \leq S(CY)$.
  \end{enumerate}
\end{lemma}

\begin{proof}
1. The definition of $I_\epsilon(\omega)$ directly implies that it 
is a non-decreasing function of $\epsilon$.

2. To prove the concavity, let $V_1:AC \rightarrow \hat{A}\hat{C}W$ and 
$V_2:AC \rightarrow \hat{A}\hat{C}W$ be the isometries attaining the 
maximum for $\epsilon_1$ and $\epsilon_2$, respectively, which act as 
follows:
\[
    V_1 \ket{\psi_x}^A\ket{\sigma_x}^C=\ket{\xi_x}^{\hat{A}\hat{C}W} \text{ and } 
    V_2 \ket{\psi_x}^A\ket{\sigma_x}^C=\ket{\zeta_x}^{\hat{A}\hat{C}W}.   
\]
For $0\leq \lambda \leq 1$,
define the isometry $U:AC \rightarrow \hat{A}\hat{C}WRR'$ by letting, for all $x$,
\[
  U\!\ket{\psi_x}^A\!\ket{\sigma_x}^C 
    \!:=\!\sqrt{\!\lambda}\ket{\xi_x}^{\hat{A}\hat{C}W}\!\!\ket{00}^{RR'}\!\!\!+\!\!\sqrt{1\!-\!\lambda}\ket{\zeta_x}^{\hat{A}\hat{C}W}\!\!\ket{11}^{RR'},
\]
where systems $R$ and $R'$ are qubits. 
Then, the reduced state on the systems $X\hat{A}\hat{C}$ is 
$\tau^{X\hat{A}\hat{C}}=\sum_x p(x) \proj{x}^X\otimes \tau_x^{\hat{A}\hat{C}}$, 
where $\tau_x^{\hat{A}\hat{C}}=\lambda \xi_x^{\hat{A}\hat{C}}+(1-\lambda)\zeta_x^{\hat{A}\hat{C}}$; 
therefore, the fidelity is bounded as follows:
\begin{align*}
  F(\omega^{XA\hat{C}}\!,\tau^{X\hat{A}\hat{C}}) 
    &= \sum_x p(x) \sqrt{\bra{\psi_x}
                          \left(\lambda \xi_x^{\hat{A}\hat{C}}+(1-\lambda)\zeta_x^{\hat{A}\hat{C}}\right)
                         \ket{\psi_x}} \\
    &\!\!\!\!\!\!\!\!\!\!\!\!\!\!\!\!\!\!\!\!\!\!\!\!\!\!\!\!\!\!\!\!\!\!\!\!\!\!\!\!\!
     \geq \lambda \sum_x p(x) \sqrt{\!\bra{\psi_x} \xi_x^{\hat{A}\hat{C}} \ket{\psi_x}\!}
          + (1-\lambda)\sum_x p(x) \sqrt{\!\bra{\psi_x}\zeta_x^{\hat{A}\hat{C}} \ket{\psi_x}\!} \\
    &\!\!\!\!\!\!\!\!\!\!\!\!\!\!\!\!\!\!\!\!\!\!\!\!\!\!\!\!\!\!\!\!\!\!\!\!\!\!\!\!\!
     \geq 1-\left( \lambda\epsilon_1 +(1-\lambda)\epsilon_2 \right),
\end{align*}
where the second line follows from the concavity of the function $\sqrt{x}$, 
and the last line follows by the definition of the isometries $V_1$ and $V_2$.
Now, define $W':=WRR'$ and let $\epsilon=\lambda\epsilon_1 +(1-\lambda)\epsilon_2$. 
According to Definition \ref{def:I_epsilon}, we obtain
\begin{align*}
  I_\epsilon(\omega) &\geq I(X:\hat{C}W')_{\tau}\\ 
                     &=    I(X:R)_{\tau}+I(X:\hat{C}W|R)_{\tau}+I(X:R'|\hat{C}WR)_{\tau}\\
                     &\geq I(X:\hat{C}W|R)_{\tau}
                      =    \lambda I_{\epsilon_1}(\omega)+(1-\lambda)I_{\epsilon_2}(\omega),
\end{align*}
where the third line is due to strong subadditivity of the quantum mutual information.

3. The function is non-decreasing and concave for $\epsilon \geq 0 $, so it is continuous 
for $\epsilon > 0 $. 
The concavity implies furthermore that $I_{\epsilon}$ is lower semi-continuous at 
$\epsilon=0$. On the other hand, since the fidelity and mutual information are 
both continuous functions of CPTP maps, and the domain of the optimization is a 
compact set, we conclude that $I_\epsilon(\omega)$ is also upper semi-continuous at 
$\epsilon=0$, so it is continuous at $\epsilon=0$ \cite[Thms.~10.1,~10.2]{Rockafeller}. 

4. In the definition of $I_{\epsilon}(\omega_1 \otimes \omega_2)$, let the isometry 
$V_0:A_1 C_1 A_2C_2 \rightarrow \hat{A}_1\hat{C}_1\hat{A}_2\hat{C}_2 W$ be the 
one attaining the maximum which acts on the purified source state with purifying 
systems $X_1'$ and $X_2'$ as follows: 
\begin{align*}
  \ket{\xi}&^{X_1\!X_1'\!X_2\!X_2'\!\hat{A}_1\!\hat{C}_1\!\hat{A}_2\!\hat{C}_2\!W} \\
           &\phantom{====}
            =(\1_{X_1\!X_1'\!X_2\!X_2'}\otimes V_0)\ket{\omega_1}^{X_1\!X'_1\!A_1\!C_1}
                                                   \ket{\omega_1}^{X_2\!X'_2\!A_2\!C_2}\!\!.
\end{align*}
Now, define the isometry $V_1:A_1 C_1 \rightarrow \hat{A}_1\hat{C}_1\hat{A}_2 \hat{C}_2W X_2X'_2$ 
acting only on the systems $A_1 C_1$ with the output state $\hat{A}_1\hat{C}_1$ and the 
environment $W_1:=\hat{A}_2\hat{C}_2 W X_2X'_2$ as follows:
\[
  \ket{\xi}^{X_1\!X_1'\!X_2\!X_2'\!\hat{A}_1\!\hat{C}_1\!\hat{A}_2\!\hat{C}_2\!W}
                           = (\1_{X_1X_1'}\otimes V_1)\ket{\omega_1}^{X_1X_1'A_1C_1}. 
\]
Hence, we obtain
\begin{align*}
  F(\omega_1^{X_1\!A_1\!C_1}\!,\xi^{X_1\!\hat{A}_1\!\hat{C}_1})
                             &\geq F\!\left(\omega_1^{X_1\!A_1\!C_1}\!\otimes\!\omega_2^{X_2\!A_2\!C_2}\!,
                                       \xi^{X_1\!X_2\!\hat{A}_1\!\hat{C}_1\!\hat{A}_2\!\hat{C}_2}\!\right) \\
                             &\geq 1 - \epsilon, 
\end{align*}
where the first inequality is due to monotonicity of the fidelity under 
CPTP maps, and the second inequality follows by the definition of $V_0$. 
Consider the isometry 
$V_2:A_2 C_2 \rightarrow \hat{A}_1\hat{C}_1\hat{A}_2 \hat{C}_2W X_1X'_1$ 
defined in a similar way, with the output state $\hat{A}_2\hat{C}_2$ and 
the environment $W_2:=\hat{A}_1\hat{C}_1 W X_1X'_1$. 
Therefore, we obtain
\begin{align*}
    I_{\epsilon}(\omega_1) +I_{\epsilon}(\omega_2) &\geq I(X_1:\hat{C}_1W_1)+I(X_2:\hat{C}_2W_2)\\
    &\geq I(X_1:\hat{C}_1\hat{C}_2W)+I(X_2:\hat{C}_1\hat{C}_2WX_1)\\
    &= I(X_1X_2:\hat{C}_1\hat{C}_2W)
     = I_{\epsilon}(\omega_1 \otimes \omega_2),
\end{align*}
where the second line is due to data processing. 

5. In the definition of $I_0(\omega)$ let $V_0:AC \rightarrow \hat{A}\hat{C}W$ be the isometry attaining the maximum with $F(\omega^{XAC},\xi^{X\hat{A}\hat{C}})=1$. Hence, we obtain
\begin{align*}
  I_0(\omega)&= I(X:\hat{C}W)
       =    I(XY:\hat{C}W) \\
      &=    I(Y:\hat{C}W)+I(X:\hat{C}W|Y) \\
      &\leq S(Y)+I(X:\hat{C}W|Y) \\
      &=    S(Y)+I(X:W|Y)+I(X:\hat{C}|WY) \\
      &\leq S(Y)+I(X:W|Y)+S(C|WY)\\
      &\leq S(Y)+I(X:W|Y)+S(C|Y),
\end{align*}
where the first line follows because $Y$ is a function of $X$. The second and fourth
line are due to the chain rule. The third line follows because for the classical 
system $Y$ the conditional entropy $S(Y|\hat{C}W)$ is non-negative. The penultimate 
line follows because for any $x$ the state on the system $\hat{C}$ is pure. The 
last line is due to strong sub-additivity of the entropy. Furthermore, for every 
$y$, the ensemble $\cE_y$ is irreducible; hence, the conditional mutual information 
$I(X:W|Y)=0$ which follows from the detailed discussion on page 2028 of \cite{Barnum2001_2}. 
%
%The fact that the fidelity is equal to zero implies that the isometry $V:AC \to \hat{A}W$ acts as $V \ket{\psi_x}^A\ket{\sigma_x}^C=\ket{\psi_x}^{\hat{A}}\ket{w_x}^W$ which implies the following
%\begin{align*}
%   \bra{\psi_{x'}}^A\bra{\sigma_{x'}}^C V^{\dagger}V \ket{\psi_x}^A\ket{\sigma_x}^C&=\braket{\psi_{x'}}{\psi_{x}}\braket{\sigma_{x'}}{\sigma_{x}}\\
%   &=\braket{\psi_{x'}}{\psi_{x}}\braket{w_{x'}}{w_{x}}.
%\end{align*}
%If $x\in \mathcal{X}_y$ and $x'\in \mathcal{X}_{y'}$ for $y\neq y'$, then the above equation is equal to zero. However, if $x, x'\in \mathcal{X}_y$, we can conclude that
%$\braket{\sigma_{x'}}{\sigma_x}=\braket{w_{x'}}{w_x}$;
%
%this implies that for a given $y$ and all $x\in \mathcal{X}_y$ there is an isometry $U:C \rightarrow W$ such that $U \ket{\sigma_x}^C=\ket{w_x}^W$. Hence, we obtain
%\begin{align*}
% &\sum_{x\in \mathcal{X}_y} p(x|y) \proj{x}^X \otimes \proj{w_x}^W \\
%& \>\>\>\>\>=(\1_X \otimes U) \sum_{x\in \mathcal{X}_y} p(x|y) \proj{x}^X \otimes \proj{\sigma_x}^C  (\1_X \otimes U^{\dagger}),   
%\end{align*}
%which implies that $I(X:W|Y)=I(X:C|Y)$.
%
%For a reducible source, the isometry $V:AC \to \hat{A}W$ can depend on $y$ because the identity of the orthogonal subspaces can be determined without perturbing the state. The fact that the fidelity is equal to zero implies that for any given $y$, $V_y \ket{\psi_x}^A\ket{\sigma_x}^C=\ket{\psi_x}^{\hat{A}}\ket{w_x}^W$, and therefore $\braket{\sigma_{x'}}{\sigma_x}=\braket{w_{x'}}{w_x}$. 
\end{proof}

\begin{proof-of}[{of the converse part of Theorem \ref{theorem: main}}]
We start by observing 
\[
  nQ+S(B_0) \geq S(C_A)+S(B_0)
            \geq S(C_A B_0)    
            =    S(\hat{A}^n W_B),
\]
where the second inequality is due to subadditivity of the entropy, 
and the equality follows because the decoding isometry $V_B$ does not change 
the entropy. Hence, we get 
\begin{align}
  \label{eq:converse Schumacher assisted 1}
  nQ+S(B_0) &\geq S(\hat{A}^n)+S(W_B|\hat{A}^n)            \nonumber\\
    &\geq S(\hat{A}^n)+S(W_B|\hat{A}^n X^n)                \nonumber\\
    &\geq S(A^n)+S(W_B|\hat{A}^n X^n)-n \delta(n,\epsilon) \nonumber\\
    &=    S(A^n) \!+\! S(\hat{A}^nW_B| X^n\!) \!-\! S(\hat{A}^n| X^n)\!-\!n \delta(n,\epsilon)\nonumber\\
    &=    S(A^n) \!+\! S(\hat{C}^nW_A| X^n) \!-\! S(\hat{A}^n| X^n) \!-\! n \delta(n,\epsilon)\nonumber\\
    &\geq S(A^n)+S(\hat{C}^nW_A| X^n)- 3n \delta(n,\epsilon),
\end{align}
where in the first and second line we use the chain rule and subadditivity of entropy.
The inequality in the third line follows from the decodability of the system $A^n$:
the fidelity criterion (\ref{eq:Schumcaher assisted fidelity}) implies that the 
output state on systems $\hat{A}^n$ is $2\sqrt{2\epsilon}$-close to the 
original state $A^n$ in trace norm; then apply the Fannes-Audenaert inequality 
\cite{Fannes1973,Audenaert2007} where 
$\delta(n,\epsilon)=\sqrt{2\epsilon} \log|A| + \frac1n h(\sqrt{2\epsilon})$.   
The equalities in the fourth and the fifth line are due to the chain rule 
and the fact that  for any $x^n$ the overall state of $\hat{A}^n\hat{C}^nW_AW_B$ is pure.
In the last line, we use the decodability of the systems $X^nA^n$, that is the 
output state on systems $X^n\hat{A}^n$ is $2\sqrt{2\epsilon}$-close to the 
original states $X^nA^n$ in trace norm, then we apply the Alicki-Fannes 
inequality \cite{Alicki2004,Winter2016}. 

Moreover, we bound $Q$ as follows:
\begin{align}\label{eq:converse Schumacher assisted 2}
  nQ &\geq S(C_A)                     
      \geq S(C_A|\hat{C}^nW_A)           \nonumber \\
     &=    S(A^nC^nA_0) -S(\hat{C}^nW_A) \nonumber \\
     &=    S(A^nC^nY^n)+S(A_0) -S(\hat{C}^nW_A),
\end{align}
where the first equality follows because the encoding isometry 
$V_A:A^nC^nA_0 \rightarrow C_A\hat{C}^nW_A$ does not the change the entropy. 
Adding Eqs. (\ref{eq:converse Schumacher assisted 1}) and 
(\ref{eq:converse Schumacher assisted 2}), we thus obtain
\begin{align}\label{eq:converse Schumacher}
  Q &\geq \frac{1}{2}(S(A)+S(ACY))-\frac{1}{2n}I(\hat{C}^nW_A:X^n)-\frac{3}{2}\delta(n,\epsilon)\nonumber \\
  &\geq \frac{1}{2}(S(A)+S(ACY))-\frac{1}{2n}I(\hat{C}^nW_AW_B:X^n)-\frac{3}{2}\delta(n,\epsilon)\nonumber \\
  &\geq \frac{1}{2}(S(A)+S(ACY))-\frac{1}{2n} I_\epsilon(\omega^{\otimes n})
                                -\frac{3}{2}\delta(n,\epsilon) \nonumber \\
  &\geq \frac{1}{2}(S(A)+S(ACY))-\frac{1}{2} I_\epsilon(\omega)-\frac{3}{2}\delta(n,\epsilon) \nonumber 
\end{align}
where the second line is due to data processing. The third line follows from 
Definition \ref{def:I_epsilon}. The last line follows from point 4 of Lemma 
\ref{lemma:I_epsilon properties}. In the limit of $\epsilon \to 0$ and 
$n \to \infty $, the rate is bounded by
\begin{align*}
  Q &\geq \frac{1}{2}(S(A)+S(ACY))-\frac{1}{2} I_0(\omega)\\
    &\geq \frac{1}{2}(S(A)+S(ACY))-\frac{1}{2}S(CY)\\
    &=    \frac{1}{2}(S(A)+S(A|CY)),
\end{align*}
where the first line follows from point 3 of Lemma \ref{lemma:I_epsilon properties}
stating that $I_\epsilon(\omega)$ is continuous at $\epsilon=0$. 
The second line is due to point 5 of Lemma \ref{lemma:I_epsilon properties}. 
\end{proof-of}

\section{Complete rate region}
In this section, we find the complete rate region of achievable rate pairs $(E,Q)$.

\begin{theorem}
  \label{theorem:complete rate region}
  For the source $\omega^{XACY}$, all asymptotically achievable entanglement and 
  quantum rate pairs $(E,Q)$ satisfy
  \begin{align*}
    Q   &\geq \frac{1}{2}(S(A)+S(A|CY)),\\
    Q+E &\geq  S(A). 
  \end{align*}
  Conversely, all the rate pairs satisfying the above inequalities are achievable.
\end{theorem}
\begin{proof}
The first inequality comes from Theorem \ref{theorem: main}. 
For the second inequality, consider any code with quantum communication
rate $R$ and entanglement rate $E$. By using an additional communication
rate $E$, Alice and Bob can distribute the entanglement first, and then
apply the given code, converting it into one without preshared 
entanglement and communication rate $Q+E$, having exactly the same
fidelity. By Remark \ref{remark:E=0}, $Q+E \geq S(A)$.

As for the achievability, the corner point 
$(\frac{1}{2} I(A:CY),\frac{1}{2}(S(A)+S(A|CY)))$ is achievable, 
because QSR which is used as the achievability protocol in 
Theorem \ref{theorem: main} uses $\frac{1}{2} I(A:CY)$ ebits of 
entanglement between Alice and Bob. 
Furthermore, all the points on the line $Q+E = S(A)$ for 
$Q \geq \frac{1}{2}(S(A)+S(A|CY))$ are achievable because one 
ebit can be distributed by sending a qubit. 
All other rate pairs are achievable by resource wasting. The rate region is depicted in
Fig.~\ref{fig:E-Q}
\end{proof}

\begin{figure}[ht] 
  \includegraphics[width=0.49\textwidth]{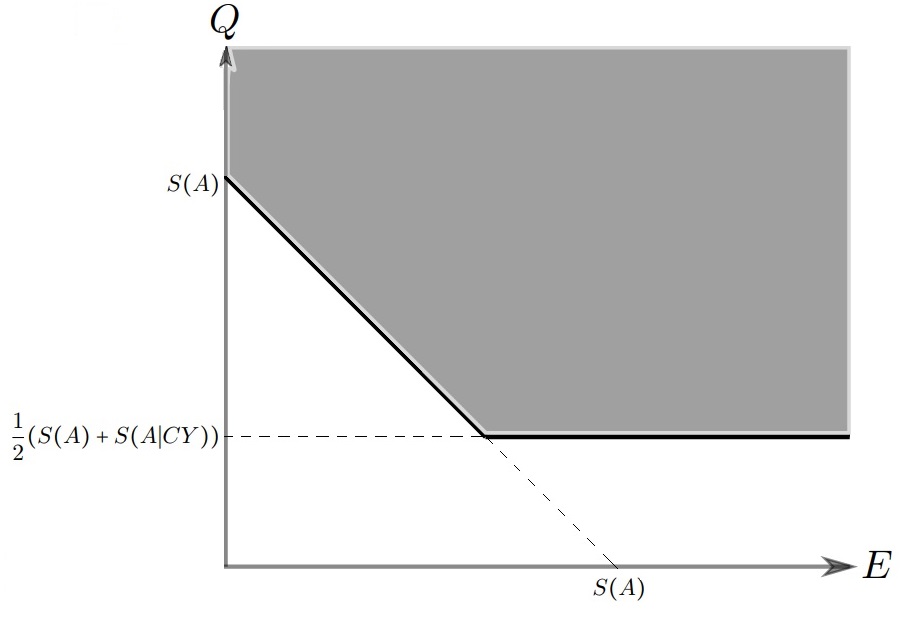}
  \caption{The optimal rate region of quantum and entanglement rates.}
  \label{fig:E-Q}
\end{figure}

\section{Discussion}
First of all, let us look what our result tell us in the cases of
blind and visible compression. 

\begin{corollary}
  \label{corollary:blind}
  In blind compression (i.e. if $C$ is trivial, or more generally the
  states $\sigma_x$ are all identical), the compression of the source 
  $\omega^{XACY}$ reduces to the entanglement-assisted Schumacher 
  compression for which Theorem \ref{theorem: main} gives the optimal asymptotic quantum rate
  \[
%    \phantom{=======}
    Q = \frac{1}{2}(S(A)+S(A|Y))=S(A)-\frac{1}{2}S(Y).
%    \phantom{=======}\blacksquare
  \]
  This implies that if the source is irreducible, then this rate is equal to the
  Schumacher limit $S(A)$. In other words, the entanglement does not help the 
  compression. Moreover, due to Theorem \ref{theorem:complete rate region}, a 
  rate $\frac{1}{2}S(Y)$ of entanglement is consumed in the compression,
  and $E+Q\geq S(A)$ in general.
\qed
\end{corollary}

\medskip
The blind compression of a source $\omega^{XAY}$ is also considered 
in \cite{Barnum2001_2}, but there instead of entanglement, a noiseless classical
channel was assumed in addition to the quantum channel.
It was shown that the optimal quantum rate assisted with free classical communication 
is equal to $S(A)-S(Y)$, while a rate $S(Y)$ of classical communication suffices.
By sending the classical information using dense coding \cite{Bennett1992},
spending $\frac12$ ebit and $\frac12$ qubit per cbit, we can recover the 
quantum and entanglement rates of Corollary \ref{corollary:blind}.
This means that our converse implies the optimality of the quantum rate
from \cite{Barnum2001_2}.
%which is smaller than the optimal quantum rate assisted 
%with free entanglement derived in Corollary \ref{corollary:blind}. 
%However, for an irreducible source, again Schumacher compression at rate $S(A)$ is optimal.

Thus we are motivated to look at a modified compression model where the resources
used are classical communication and entanglement. Namely, we let Alice and Bob 
share entanglement at rate $E$ and use classical communication at rate $C$, but
otherwise the objective is the same as in Section \ref{sec:Compression assisted by entanglement};
define the rate region as the set of all asymptotic achievable classical 
communication and entanglement rate pairs $(C,E)$, such that the decoding fidelity
asymptotically converges to $1$.  

\begin{theorem}
For a source $\omega^{XAY}$, a rate pair $(C,E)$ is achievable if and only if
\begin{align*}
  C \geq 2S(A)-S(Y),\ 
  E \geq  S(A)-S(Y). 
\end{align*}
\end{theorem}

\begin{proof}
We start with the converse. 
The first inequality follows from Theorem \ref{theorem: main}, because with 
unlimited entanglement shared between Alice and Bob, 
$\frac{1}{2}(S(A)+S(A|Y))=S(A)-\frac{1}{2}S(Y)$ qubits of quantum 
communication is equivalent to $2S(A)-S(Y)$ bits of classical communication 
due to teleportation \cite{Bennett1993} and dense coding \cite{Bennett1992}. 
The second inequality follows from \cite{Barnum2001_2}, because with free 
classical communication, the quantum rate is lower bounded by $S(A)-S(Y)$ 
which, due to super dense coding \cite{Bennett1992}, is equivalent to sharing 
$S(A)-S(Y)$ ebits when classical communication is for free. 

The achievability of the corner point $(2S(A)-S(Y),S(A)-S(Y))$ 
follows from \cite{Barnum2001_2} because the compression protocol 
uses $S(A)-S(Y)$ qubits and $S(Y)$ bits of classical communication 
which is equivalent to using $S(A)-S(Y)$ ebits of entanglement and 
$2S(A)-2S(Y)+S(Y)$ bits of classical communication, due to dense coding \cite{Bennett1992}.
Other rate pairs are achievable by resource wasting. The rate region is depicted in Fig.~\ref{fig:C-E}. 
\end{proof}

\begin{figure}[ht] 
  \includegraphics[width=0.49\textwidth]{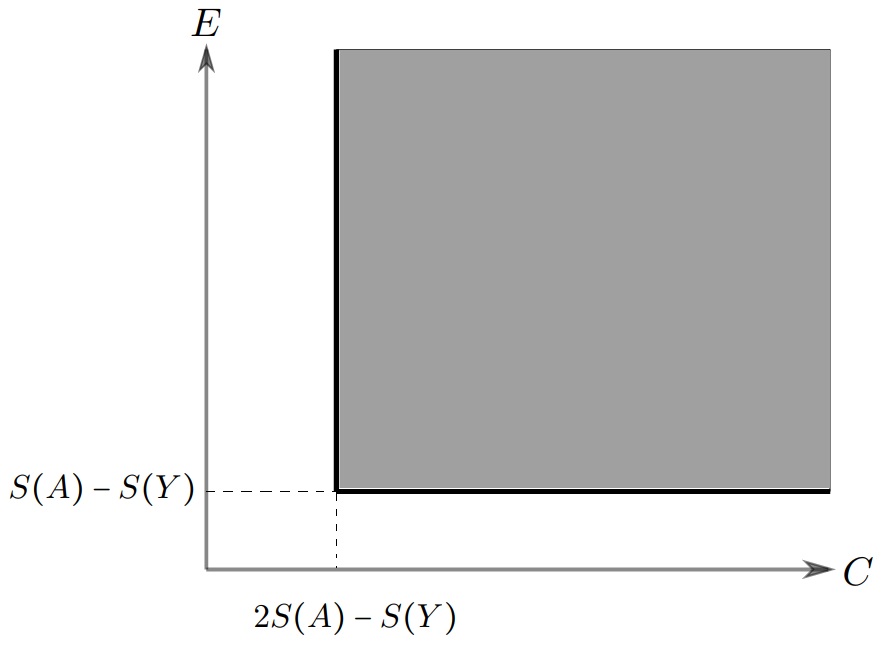} 
  \caption{The optimal rate region of classical and entanglement rates.}
  \label{fig:C-E}
\end{figure}

\begin{corollary}
\label{cor:visible}
In the visible case, our compression problem reduces to the visible version of 
Schumacher compression with entanglement assistance. In this case, 
according to Theorem \ref{theorem: main} the optimal asymptotic quantum 
rate is $Q=\frac{1}{2}S(A)$.
Moreover, a rate $E=\frac{1}{2}S(A)$ of entanglement
is consumed in the compression scheme, and $E+Q\geq S(A)$ in general.
\qed
\end{corollary}

We remark that the visible compression assisted by unlimited entanglement is 
also a special case of remote state preparation considered in \cite{Bennett2005},
from which we know that the rate $Q=\frac12 S(A)$ is achievable and optimal.
%The case of rate-bounded entanglement has been investigated in \cite{DevetakBerger2001},

The visible analogue of \cite{Barnum2001_2}, of compression using
qubit and cbit resources, was treated in \cite{Hayden2002}, where the
achievable region was determined as the union of all all pairs $(C,Q)$ such
that $Q\geq S(A|Z)$ and $C\geq I(X:Z)$, for any random
variable $Z$ forming a Markov chain $Z$---$X$---$A$. Compare to the
complicated boundary of this region the much simpler one of Corollary \ref{cor:visible},
which consists of two straight lines.

\medskip
We close by discussing several open questions for future work: 
First, the final discussion of different pairs of resources to compress suggests
that an interesting target would be the characterisation of the full triple resource
tradeoff region for $Q$, $C$ and $E$ together.

Secondly, we recall that our definition of successful decoding included 
preservation of the side information $\sigma_x^C$ with high fidelity. What is the
optimal compression rate $Q$ if the side information does not have to be preserved? 
For an example where this change has a dramatic effect on the optimal communication
rate, consider the ensemble $\cE$ consisting of the three two-qubit states
$\ket{0}^A\ket{0}^C$, $\ket{1}^A\ket{0}^C$ and $\ket{+}^A\ket{+}^C$
(where $\ket{+}=\frac{1}{\sqrt{2}}(\ket{0}+\ket{1})$), 
with probabilities $\frac12-t$, $\frac12-t$ and $2t$, respectively.
Note that $\cE$ is irreducible, hence for $t\approx 0$, we get an optimal
quantum rate of $Q\approx 1$, because $S(A) \approx S(A|C) \approx 1$.
However, by applying a CNOT unitary (with $A$ as control and $C$ as target),
the ensemble is transformed into $\cE'$ consisting of the states
$\ket{0}^A\ket{0}^{C'}$, $\ket{1}^A\ket{1}^{C'}$ and $\ket{+}^A\ket{+}^{C'}$.
The state of $A$ is not changed, only the side information, which is why we 
denote it $C'$. Hence we can apply Theorem \ref{theorem: main} to get a
quantum rate $Q\approx\frac12$, because $S(A) \approx 1$, $S(A|C) \approx 0$.

Thirdly, note that the lower bound $Q+E\geq S(A)$ in Theorem \ref{theorem:complete rate region}
holds with a strong converse (see the proof and \cite{Winter1999}).
But does $Q\geq \frac12 (S(A)+S(A|CY))$ hold as a strong converse rate
with unlimited entanglement? Likewise, in the setting of \cite{Barnum2001_2}
with unlimited classical communication, is $Q\geq S(A)-S(Y)$ a strong converse 
bound for the quantum rate?

%%%%%%%%%%%%%%%%%%%%%%%%%%%%%%%%%%%%%%%%%%%%%%%%%%%%%%%%%%%%%%%%%%%%%%%%%%%%%%%%%%%%%%%%%%%%%%%

%\section*{Acknowledgments}
%The authors acknowledge financial support from the Spanish MINECO (FIS2016-80681-P, 
%FISICATEAMO no. FIS2016-79508-P, SEVERO OCHOA no.~SEV-2015-0522, FPI), 
%the European Social Fund, the Fundaci\'o Cellex, the Generalitat de Catalunya 
%(AGAUR grants no.~2017-SGR-1127, 2017-SGR-1341 and CERCA/Program), 
%ERC AdG OSYRIS, ERC AdG IRQUAT, EU FETPRO QUIC, 
%and the National Science Centre, Poland-Symfonia grant no.~2016/20/W/ST4/00314. 

%%%%%%
%% To balance the columns at the last page of the paper use this
%% command:
%%
%\enlargethispage{-1.2cm} 
%%
%% If the balancing should occur in the middle of the references, use
%% the following trigger:
%%
%\IEEEtriggeratref{3}
%%
%% which triggers a \newpage (i.e., new column) just before the given
%% reference number. Note that you need to adapt this if you modify
%% the paper.  The "triggered" command can be changed if desired:
%%
%\IEEEtriggercmd{\enlargethispage{-20cm}}
%%
%%%%%%

%%%%%%
%% References:
%% We recommend the usage of BibTeX:
%%
%\bibliographystyle{IEEEtran}
%\bibliography{definitions,bibliofile}
%%
%% where we here have assume the existence of the files
%% definitions.bib and bibliofile.bib.
%% BibTeX documentation can be obtained at:
%% http://www.ctan.org/tex-archive/biblio/bibtex/contrib/doc/
%%%%%%

%% Or you use manual references (pay attention to consistency and the formatting style!):

%\vfill

\bibliographystyle{IEEEtran}

\end{document}